\documentclass[12pt]{article}
\usepackage{mathtools}
\usepackage[english]{babel}
\usepackage{amsthm}
\usepackage{amssymb}
\usepackage{amsmath}
\usepackage{amsmath,amssymb}
\newcommand\scalemath[2]{\scalebox{#1}{\mbox{\ensuremath{\displaystyle #2}}}}
\theoremstyle{definition}
\newtheorem{theorem}{Theorem}[section]

\date{}

\title{A Novel Catastrophic Condition for Periodically Time-varying Convolutional Encoders Based on Time-varying Equivalent Convolutional Encoders}

\author{Fan Jiang}

\begin{document}

\maketitle

\begin{abstract}
A convolutional encoder is said to be catastrophic if it maps an information sequence of infinite weight into a code sequence of finite weight. 
As a consequence of this mapping, a finite number of channel errors may cause an infinite number of information bit errors when decoding. 
This situation should be avoided.
A catastrophic condition to determine if a time-invariant convolutional encoder is catastrophic or not is stated in \cite{Massey:LSC}.
Palazzo developed this condition for periodically time-varying convolutional encoders in \cite{Palazzo:Analysis}.
Since Palazzo's condition is based on the state transition table of the constituent encoders, its complexity increases exponentially with the number of memory elements in the encoders. 
A novel catastrophic condition making use of time-varying equivalent convolutional encoders  is presented in this letter.
A technique to convert a catastrophic periodically time-varying convolutional encoder into a non-catastrophic one can also be developed based on these encoders.
Since they do not involve the state transitions of the convolutional encoder, the time complexity of these methods grows linearly with the encoder memory.
\end{abstract}

\section{Introduction}

As one of the two major types of error-correcting codes, convolutional codes are widely used in modern digital and wireless communications, e.g., IEEE 802.11.
Compared with block codes, convolutional codes have the advantages of easy encoder implementation, readiness for sequential and soft-decision decoding, etc.
A rate  $k/n$, where k and n are small positive integers with $k < n$, memory $m$ binary convolutional encoder is a linear causal finite-state sequential circuit consisting of $k$ shift-registers each with at most $m$ memory elements \cite{Forney:CCI}.
The encoder takes in an input sequence of  $k$-tuples of 0's and 1's and encode them into an output or codeword sequences of $n$-tuples by discrete-time convolution.
For any linear system, time-domain convolution may be converted into more convenient transform-domain polynomial multiplication.
As a result, rate $k/n$ convolutional encoders are often represented by $k\times n$ transfer function matrices ${\mathbf G}(D)$, where the indeterminate $D$ stands for "delay". 
For example, the matrix
\begin{equation}
{\mathbf G}(D) =
\begin{bmatrix}
1+D, & 1+D^2
\end{bmatrix}
\label{eq:cata_encoder}
\end{equation}
represents a rate $1/2$ convolutional encoder of memory 2.  
In fact, the polynomials in ${\mathbf G}(D)$ identify a feedforward realization of the encoder using memory elements, XOR gates and binary scalors \cite{Massey:LSC}.
For the encoder described by (\ref{eq:cata_encoder}), an all-one input sequence of infinite length which can be expressed as ${\mathbf I}(D) =1/(1+D)$ would generate a codeword sequence  
\[
{\mathbf C}(D)={\mathbf I}(D){\mathbf G}(D)
\] 
with only three 1's \cite{Massey:LSC}. 
What this entails is that, once this codeword sequence is transmitted over a noisy channel,  it is possible for only three channel errors to cause an infinite amount of bit errors in the input information sequence.
We call this kind of encoder catastrophic and catastrophic convolutional encoders should be avoided for obvious reasons. 
Massey et. al. showed that a convolutional encoder is non-catastrophic if and only if the greatest common divisor (GCD) of all the $n\choose k$ minors of order $k$ in ${\mathbf G}(D)$ is $D^l$ for some integer $l\geq 0$  \cite{Massey:LSC}.
Another way to tell if a convolutional encoder is catastrophic or not is to check its  state transition diagram.
The encoder is catastrophic if and only if there exists a loop, excluding the self-loop at the all-zero state, with all-zero codeword symbols on its state transition edges \cite{Lin:ECC}. 

A convolution encoder can be time-invariant to have a fixed ${\mathbf G}(D)$ or time-varying to have a ${\mathbf G}(D)$ that varies with time periodically.
Since ${\mathbf G}(D)$ is not fixed, Palazzo used the state transition diagram to determine if a periodically time-varying convolutional encoder is catastrophic or not.  
Since the number of states in this diagram increases exponentially with memory $m$, this method may become impractical when $m$ is large. 
In this letter, we develop a novel catastrophic condition to decide if a periodically time-varying convolutional encoder is catastrophic or not.
This method makes use of the {\em time-varying equivalent convolutional encoder} (TVECE), which is defined to be a time-invariant convolutional encoder that is equivalent to the periodically time-varying convolutional encoder.
Since it does not depend on the state transition diagram, the complexity of this catastrophic condition increases much slower with encoder memory.

In Section I, we first prove that every periodically time-varying convolutional encoders is equivalent to a TVECE.
Based on the TVECEs, in Section II, a catastrophic condition for periodically time-varying convolutional encoder and a method to convert a catastrophic encoder into a noncatastrophic one are derived. 
Conclusions and possible future directions for this work are presented in Section III. 

\section{Time-Varying Equivalent Convolutional Encoders}

As defined in \cite{Forney:CCI} \cite{Piret:CC}, two convolutional encoders $E$ and $E'$ of the same rate $k/n$ are equivalent if they generate the same code. 
According to this definition, a convolutional codes can have many different convolutional encoders.
These equivalent encoders generate the same code, but the mapping from the input sequences to the codeword sequences defined by them may be different.
We are going to show that every period $p$,  rate $k/n$, memory $m$, denoted as $(p, n, k, m)$,  periodically time-varying convolutional encoder is equivalent to a rate $\frac{kp}{np}$, memory $\lceil \frac{m}{p}\rceil$ time-invariant convolutional encoder in the strict sense that they both generate the same code with the same mapping.
The fact that these two convolutional encoders are equivalent in the strict sense is critical in developing the new catastrophic condition.

According to \cite{Palazzo:Analysis}, a $(p, n, k, m)$ periodically time-varying convolutional encoder can be described in two ways: via the serial and parallel description.
In the serial description, it is seen as a linear sequential circuit with time-varying connections from the memory elements of the $k$ shift-registers to the $n$ output terminals. 
At each encoding epoch, these connections may be represented by a $k\times n$  transfer function matrix of a certain time-invariant encoder called the constituent encoder. 
For a period of $p$, we have $p$ such constituent encoders.
However, in the parallel description, the input sequence of $k$-tuples is fed into all $p$ constituent encoders simultaneously, i.e., all constituent encoders encode the same input $k$-tuples at every encoding epoch. 
The time-varying nature of the encoder is manifested by periodically taking one output $n$-tuple from one of these $p$ constituent encoders as the current codeword for the encoding epoch and all the rest $p-1$ $n$-tuples can be seen as being punctured. 
We are going to use the parallel description to prove theorem 1. 

\begin{theorem} Let $E$ be a $(p, n, k, m)$ periodically time-varying convolutional encoder. 
Then $E$ is equivalent to a  rate $\frac{kp}{np}$, memory $\lceil \frac{m}{p}\rceil$ time-invariant convolutional encoder $E'$ in the sense that the same input sequence to both $E$ and $E'$ generate the same codeword sequences. 
\end{theorem}
 \begin{proof}
For the $i^{th}$ rate $k/n$ time-invariant constituent encoder, the discrete-time convolutional may be represented by a semi-infinite generator matrix, ${\mathbf G}_i$, as 
\[
{\mathbf G}_i=\left\lceil 
\begin{tabular}{ccccccccc}
$G_{i}^{0}$    & $G_i^1$  & $\cdots$ 	& $G_i^m$ & ${\mathbf 0}$ & ${\mathbf 0}$   & ${\mathbf 0}$ & $\cdots$\\
${\mathbf 0}$ & $G_i^0$  & $G_i^1$  	& $\cdots$     & $G_i^m$ 	   & ${\mathbf 0}$   & ${\mathbf 0}$  & $\cdots$\\
${\mathbf 0}$ & ${\mathbf 0}$ & $G_i^0$ 	 & $G_i^1$  & $\cdots$  & $G_i^m$ 	 & ${\mathbf 0}$   & $\cdots$\\
 	         &	   	&		&$\cdots$ & 		& 		& 
\end{tabular}
\right\rceil,
\]
where each $G_i^j, j = 0,1,\cdots,m$ is a $k\times n$ matrix of binary scalors that indicates the connections from the memory elements of the $k$ shift-registers  to the $n$ outputs.   
Based on the parallel description, we can ignore the constituent encoders whose outputs are punctured and 
combine all $p$ matrices like ${\mathbf G}_i$ into one generator matrix, ${\mathbf G}$, as following
\[
\scalemath{0.75}
{
{\mathbf G}=\left\lceil 
\begin{tabular}{cccc|ccccccccc}
$G_1^0$    & $G_2^1$  	& $\cdots$  &  $\cdots$ &  $\cdots$	& $G_{\mu_p(m)}^m$ & ${\mathbf 0}$ & ${\mathbf 0}$ & ${\mathbf 0}$ & ${\mathbf 0}$ & $\cdots$   &  ${\mathbf 0}$  &  $\cdots$\\
${\mathbf 0}$    & $ G_2^0$  	& $ G_3^1 $ &  $\cdots$ &  $\cdots$	& $G_{\mu_p(m)}^{m-1}$ & $G_{\mu_p(m+1)}^{m}$ & ${\mathbf 0}$ & ${\mathbf 0}$  & ${\mathbf 0}$ & $\cdots$   & ${\mathbf 0}$   & $\cdots$ \\
${\mathbf 0}$ & ${\mathbf 0}$ & $\ddots$  	&  $\cdots$ 	& $\cdots$  & $\ddots$   	& $\ddots$	&  $\ddots$   & ${\mathbf 0}$ & ${\mathbf 0}$  & $\cdots$  & $ {\mathbf 0}$   & $\cdots$ \\
${\mathbf 0}$ & $\cdots$ 	& ${\mathbf 0}$ 	& $G_p^0$   & $\cdots$  & $\cdots$      & $\cdots$    & $\cdots$    & $\cdots$          & $G_{\mu_p(m+p-1)}^m$        &	${\mathbf 0}$   & ${\mathbf 0}$   & $\cdots$\\\hline
${\mathbf 0}$ & $\cdots$ 	& $\cdots$ 	&  ${\mathbf 0}$ & $G_1^0$ & $G_2^1$  & $\cdots$ & $\cdots$ & $\cdots$ & $\cdots$ & $G_{\mu_p(m)}^m$ & ${\mathbf 0}$ & $\cdots$ \\
 &  	& 	&   & & $\ddots$ & $\cdots$ & $\cdots$ & $\cdots$  & $\cdots$  & $\cdots$ & $ \ddots$ & $\cdots$ \\

\end{tabular}
\right\rceil
}
\]
where $\mu_p(x) = (x \mod p) + 1$. 
Examining ${\mathbf G}$, we see that, if the first $p$ rows are combined into one composite row, the remainder of ${\mathbf G}$ consists of shifts of this composite row to the right by $lp, l=1,2,\cdots,$ columns each time.
The matrix ${\mathbf G}$ can then be segmented into blocks of $p\times p$ submatrices as the one shown in the upper-left corner of ${\mathbf G}$.
Therefore, ${\mathbf G}$ is in fact the generator matrix of a time-invariant convolutional encoder.
Since each $G_i^j$ in the $p\times p$ submatrix is a $k\times n$ matrix of binary entries, the dimension of the $p\times p$ submatrix is $kp\times np$.
Thus, the rate of the time-invariant encoder ${\mathbf G}$ represents is $\frac{kp}{np}$.

Let $m^*$ be the memory of this time-invariant encoder.
To find $m^*$, note that there are a total of $m+p$ columns in each composite row of ${\mathbf G}$ when the all-zero part is excluded.  
Because each composite row like this is segmented into blocks of $p\times p$ submatrices, the memory of the time-invariant encoder is therefore  
\begin{eqnarray*}
m^* & = & \left\lceil\frac{m+p}{p}\right\rceil-1\\
& = & \left\lceil\frac{m}{p}\right\rceil.
\end{eqnarray*}
\end{proof}

\section{Catastrophic Condition for Periodically Time-varying Convolutional Encoders Based on TVECEs}

Now that every periodically time-varying convolutional encoder is equivalent in the strict sense to a time-invariant convolutional encoder, the converse is not true.
We call a time-invariant convolutional encoder that is equivalent in the strict sense to a periodically time-varying convolutional encoder a {\em time-varying equivalent convolutional encoder} (TVECE).
Utilizing TVECEs, catastrophic periodically time-varying convolutional encoders can be identified and converted to be non-catastrophic.
These results are summarized in the following theorem:
 
\begin{theorem}
Let $E$ be a period $p$, rate $k/n$ periodically time-varying convolutional encoder.
Then $E$ is catastrophic if and only if  its time-varying equivalent convolutional encoder $E'$  is catastrophic.
Moreover, if $E$ is catastrophic, it can be converted to be non-catastrophic by dividing the transfer function matrices of all its constituent encoders by $f(D^p)$, 
where $f(D)$ is the GCD of all minors of order $kp$ in the transfer function matrix of $E'$.
\end{theorem}
\begin{proof}
The first part of the theorem is straightforward.
Since a periodically time-varying convolutional encoder is equivalent to its time-varying equivalent convolutional encoder in the strict sense, 
their mapping from the input sequences to the output sequences are identical. 
Therefore, if the TVECE is catastrophic and maps an input of infinite weight to a codeword of finite weight, the periodically time-varying convolution ecnoder does the same and thus is also catastrophic and vice versa.

Let ${\mathbf G}_{E'}(D)$ be the $kp\times kn$ transfer function matrix of $E'$.
Since $E'$ is a time-invariant convolutional encoder, dividing ${\mathbf G}_{E'}(D)$ by $f(D)$ converts it into a non-catastrophic encoder \cite{Forney:CCI}. 
To prove the second half of the theorem, we need to show that this is equivalent to dividing all constituent encoders of $E$ by $f(D^p)$.

Since the encoding process  can be described by ${\mathbf C}(D)={\mathbf I}(D){\mathbf G}(D)$, dividing ${\mathbf I}(D)$ by $f(D)$ has the same effect as dividing ${\mathbf G}(D)$ by $f(D)$ on ${\mathbf C}(D)$.
Therefore, we can show equivalence of encoders just by examining ${\mathbf I}(D)$, the input sequences.
Notice that there are $p$ input sequences of $k$-tuples to $E'$. 
Let's denote these $p$ sequences as ${\mathbf I}_1(D),\cdots, {\mathbf I}_p(D)$.
According to the parallel description of periodically time-varying convolutional encoders, the single input sequence of $k$-tuples to $E$, denoted as ${\mathbf I}_E(D)$, is simply the serialization of  ${\mathbf I}_1(D),\cdots, {\mathbf I}_p(D)$.
This allows us to write \cite{Lin:ECC}
\begin{align}
{\mathbf I}_E(D)= {\mathbf I}_1(D^p) + D{\mathbf I}_2(D^p)+\cdots+D^{p-1}{\mathbf I}_p(D^p)\nonumber.
\end{align}

If we divide all $p$ input sequences to $E'$ by $f(D)$, we obtain sequences as $\frac{{\mathbf I}_1(D)}{f(D)},\cdots,\frac{{\mathbf I}_p(D)}{f(D)}$. 
Serialization of these yields the corresponding input sequence to $E$, denoted as $I'_E(D)$, as
\begin{eqnarray*}
I'_E(D) & = & \frac{{\mathbf I}_1(D^p)}{f(D^p)} + D\frac{{\mathbf I}_2(D^p)}{f(D^p)}+\cdots+D^{p-1}\frac{{\mathbf I}_p(D^p)}{f(D^p)} \\
& = & \frac{{\mathbf I}_1(D^p)+D{\mathbf I}_2(D^p)+\cdots+D^{p-1}{\mathbf I}_p(D^p)}{f(D^p)}\\
& = & \frac{{\mathbf I}_E(D)}{f(D^p)}
\end{eqnarray*}
This shows that dividing ${\mathbf G}_{E'}(D)$ by $f(D)$ is equivalent to dividing ${\mathbf G}_{E}(D)$ by $f(D^p)$ in the sense that they generate the same codewords with the same mapping between input and output.
Hence, $E$ can be converted to be non-catastrophic by dividing the transfer function matrices of all its constituent encoders by $f(D^p)$.
\end{proof}

Based on theorem 1 and 2, with the help of the TVECEs, not only can we develop a catastrophic condition for periodically time-varying convolutional encoders, 
but we can convert them to be non-catastrophic as well.  
The computational complexity of both the catastrophic condition and the conversion method are dominated by computing the GCD of polynomials with binary coefficients.
Finding this GCD is similar to finding the GCD of integers.
We can use the well-known recursive Euclid's algorithm \cite{Cormen:Algo}, whose time complexity is linear with the order $m$ of the polynomials.
Therefore, our new catastrophic condition and conversion method significantly save computational time compared with Palazzo's technique in \cite{Palazzo:Analysis}. 

\section{Conclusions}

In this letter, we showed that every periodically time-varying convolutional encoder is equivalent to a TVECE in the sense that not only do they generate the same code, they have the same mapping from the input sequences to the codeword sequences as well.
Based on this equivalence, we can make use of the TVECEs to identify catastrophic periodically time-varying convolutional encoders.
Moreover, using the GCD of all the minors of order $kp$ in ${\mathbf G}_{E'}(D)$ of the TVECE, we can convert a catastrophic periodically time-varying convolutional encoder to be non-catastrophic.
TVECEs provide us with a convenient alternative way to analyze periodically time-varying convolutional encoders. 
Future work will explore algebraic structure and properties of these encoders via TVECEs further.

\end{document}